%% file: k-mismatches.tex
\documentclass[11pt]{article}
\usepackage{fullpage}
\usepackage{amsthm,amssymb,amsmath}  
\usepackage{xspace,enumerate}
\usepackage[utf8]{inputenc}
\usepackage{thmtools}
\usepackage{thm-restate}
\usepackage{authblk}

  \theoremstyle{plain}
  \newtheorem{theorem}{Theorem}
  \newtheorem{lemma}{Lemma}  
   
  \newtheorem{proposition}[theorem]{Proposition}

  \theoremstyle{definition}
  \newtheorem{definition}{Definition}

  \newtheorem{remark}[definition]{Remark}
 
  \newtheorem*{claim}{Claim}
  \newcommand{\sample}{\mathcal{S}}

  \usepackage{verbatim}
  \usepackage{amsmath,amsfonts}
  \usepackage{tikz}
  \usetikzlibrary{trees}
  \usepackage{forest}  
  \usepackage[T1]{fontenc}
  \usepackage{lmodern}
  \usepackage[utf8]{inputenc}
  \usepackage{xspace}
  \usepackage[noend, boxed, noline]{algorithm2e}
  \usepackage[colorinlistoftodos]{todonotes}
  \usepackage[colorlinks=true, allcolors=blue]{hyperref}
  \usepackage[capitalise]{cleveref}
  \usepackage{subcaption}
  \usepackage{enumerate}
  \usepackage{microtype}
  \usepackage[]{setspace}
  \SetAlCapSkip{0.5em}
  \usetikzlibrary{arrows.meta}
  \usetikzlibrary{patterns}
  \newcommand{\W}{\mathbf{W}}
  \newcommand{\defproblem}[3]{
  \vspace{2mm}
  \noindent\fbox{
  \begin{minipage}{0.96\textwidth}
  \textsc{#1}

  \smallskip
  \noindent
  {\bf{Input:}} #2
  
  \smallskip
  \noindent
  {\bf{Output:}} #3
  \end{minipage}
  }
  \vspace{2mm}
}

  \def\dd{\mathinner{.\,.}}
  
  \newcommand{\floor}[1]{\left\lfloor #1 \right\rfloor}
  \newcommand{\ceil}[1]{\left\lceil #1 \right\rceil}
  \newcommand{\per}{\textsf{per}}
  \newcommand{\Oh}{\mathcal{O}}
  \newcommand{\cO}{\mathcal{O}}
  \newcommand{\G}{\mathcal{G}}
  \def\dd{\mathinner{.\,.}}

  \newcommand{\rot}{\mathsf{rot}}
  \newcommand{\Chain}{\mathsf{Chain}}
  \newcommand{\lcp}{\mathsf{lcp}}
  \newcommand{\LeftMis}{\mathsf{LeftMisper}}
  \newcommand{\RightMis}{\mathsf{RightMisper}}
  \newcommand{\Misp}{\mathsf{Misper}}
  \newcommand{\Mis}{\mathsf{Mis}}
  \newcommand{\lcsuf}{\mathsf{lcs}}
  \newcommand{\frag}{\mathsf{frag}}

\title{Circular Pattern Matching with $k$ Mismatches}

\author[1]{Panagiotis Charalampopoulos}
\author[2,3]{Tomasz Kociumaka}
\author[4]{Solon P. Pissis}
\author[3]{Jakub Radoszewski}
\author[3]{Wojciech Rytter}
\author[3]{Juliusz Straszyński}
\author[3]{Tomasz Waleń}
\author[3]{Wiktor Zuba}

\affil[1]{Department of Informatics, King's College London, UK\\
    \texttt{panagiotis.charalampopoulos@kcl.ac.uk}}
\affil[2]{Department of Computer Science, Bar-Ilan University, Ramat Gan, Israel}
\affil[3]{Institute of Informatics, University of Warsaw, Warsaw, Poland\\
    \texttt{$\{$kociumaka,jrad,rytter,jks,walen,w.zuba$\}$@mimuw.edu.pl}}
\affil[4]{CWI, Amsterdam, The Netherlands\\
    \texttt{solon.pissis@cwi.nl}}
\date{\vspace{-5ex}}

\begin{document}
\maketitle
\thispagestyle{empty}
\begin{abstract}
The $k$-mismatch problem consists in computing the Hamming distance between a pattern $P$ of length $m$ and
every length-$m$ substring of a text $T$ of length $n$, if this distance is no more than $k$.	
In many real-world applications, any cyclic rotation of $P$ is a relevant pattern, and thus one is interested
in computing the minimal distance of every length-$m$ substring of $T$ and any cyclic rotation of $P$.
This is the circular pattern matching with $k$ mismatches ($k$-CPM) problem.
A multitude of papers have been devoted to solving this problem but, to the best of our knowledge, only average-case upper bounds are known.
In this paper, we present the first non-trivial worst-case upper bounds for the $k$-CPM problem.
Specifically, we show an $\cO(nk)$-time algorithm and an $\cO(n+\frac{n}{m}\,{k^4})$-time algorithm.
The latter algorithm applies in an extended way a technique that was very recently developed for the $k$-mismatch problem [Bringmann et al., SODA 2019].

A preliminary version of this work appeared at FCT 2019. % with time complexity $\cO(n+\frac{n}{m}\,{k^5})$.
In this version we improve the time complexity of the main algorithm from 
$\cO(n+\frac{n}{m}\,{k^5})$ to $\cO(n+\frac{n}{m}\,{k^4})$.
\end{abstract}

\clearpage
\setcounter{page}{1}

\section{Introduction}\label{sec:intro}

Pattern matching is a fundamental problem in computer science~\cite{AlgorithmsOnStrings}.
It consists in finding all substrings of a text $T$ of length $n$ that match a pattern $P$ of length $m$.
In many real-world applications, a measure of similarity is usually introduced allowing for \emph{approximate} matches
between the given pattern and substrings of the text. The most widely-used similarity measure 
is the Hamming distance between the pattern and all length-$m$ substrings of the text.

Computing the Hamming distance between $P$ and all length-$m$ substrings of $T$ has been investigated for the past 30 years.
The first efficient solution requiring $\cO(n\sqrt{m\log m})$ time was independently developed by Abrahamson~\cite{Abrahamson} and Kosaraju~\cite{Kosaraju} in 1987.
The $k$-mismatch version of the problem asks for finding only the substrings of $T$ that are close to $P$, specifically, at Hamming distance at most $k$.
The first efficient solution to this problem running in $\cO(nk)$ time was developed in 1986 by Landau and Vishkin~\cite{DBLP:journals/tcs/LandauV86}.
It took almost 15 years for a breakthrough result by Amir et al.~improving this to $\cO(n\sqrt{k \log k})$~\cite{DBLP:journals/jal/AmirLP04}.
More recently, there has been a resurgence of interest in the $k$-mismatch problem. Clifford et al.\ gave 
an $\cO((n/m)(k^2\log k) + n \text{polylog} n)$-time algorithm~\cite{DBLP:conf/soda/CliffordFPSS16}, 
which was subsequently improved further by Gawrychowski and Uznański to $\cO((n/m)(m+k\sqrt{m})\text{polylog} n)$~\cite{DBLP:conf/icalp/GawrychowskiU18}.
In~\cite{DBLP:conf/icalp/GawrychowskiU18}, the authors have also provided evidence that any further progress in this problem is rather unlikely.

The $k$-mismatch problem has also been considered on compressed representations of the 
text~\cite{DBLP:journals/talg/BilleFG09,DBLP:conf/soda/BringmannWK19,DBLP:conf/isaac/GawrychowskiS13,DBLP:conf/stringology/Tiskin14},
in the parallel model~\cite{DBLP:journals/tcs/GalilG87}, and in the streaming model~\cite{DBLP:conf/focs/PoratP09,DBLP:conf/soda/CliffordFPSS16,DBLP:conf/soda/CliffordKP19}.
Furthermore, it has been considered in non-standard stringology models, such as the parameterized 
model~\cite{DBLP:journals/talg/HazayLS07} and the order-preserving model~\cite{DBLP:journals/tcs/GawrychowskiU16}.

In many real-world applications, such as in bioinformatics~\cite{BMCgenomicsAyad2017,DBLP:journals/almob/GrossiIMPPRV16,DBLP:books/sp/17/IliopoulosPR17,DBLP:conf/wea/BartonIKPRV15} 
or in image processing~\cite{DBLP:journals/prl/AyadBP17,DBLP:journals/prl/Palazon-Gonzalez15,DBLP:journals/pr/Palazon-GonzalezMV14,DBLP:journals/ivc/Palazon-GonzalezM12}, 
any cyclic shift (rotation) of $P$ is a relevant pattern, and thus one is interested
in computing the minimal distance of every length-$m$ substring of $T$ and any cyclic rotation of $P$, if this distance is no more than $k$.
This is the circular pattern matching with $k$ mismatches ($k$-CPM) problem.
A multitude of papers~\cite{DBLP:journals/jea/FredrikssonN04,DBLP:journals/almob/BartonIP14,DBLP:conf/bcb/AzimIRS14,DBLP:conf/isbra/AzimIRS15,DBLP:conf/lata/BartonIP15,DBLP:journals/jea/HirvolaT17} 
have thus been devoted to solving the $k$-CPM problem but, to the best of our knowledge, only average-case upper bounds are known; 
i.e.~in these works the assumption is that text $T$ is uniformly random. The main result states that, after preprocessing pattern $P$, 
the average-case optimal search time of $\cO(n\frac{k+\log m}{m})$~\cite{DBLP:conf/cpm/ChangM94} 
can be achieved for certain values of the error ratio $k/m$ (see~\cite{DBLP:conf/lata/BartonIP15,DBLP:journals/jea/FredrikssonN04} for more details on the preprocessing costs).
Note that the exact (no mismatches allowed) version of the CPM problem can be solved as fast as exact pattern matching; namely, in $\cO(n)$ time~\cite{DBLP:books/daglib/0025614}.

In this paper, we draw our motivation from {\bf (i)} the importance of the $k$-CPM problem in real-world applications 
and {\bf (ii)} the fact that no (non-trivial) worst-case upper bounds are known.
Trivial here refers to running the fastest-known algorithm for the $k$-mismatch problem~\cite{DBLP:conf/icalp/GawrychowskiU18} separately 
for each of the $m$ rotations of $P$. This yields an $\cO(n(m+k\sqrt{m})\text{polylog} n)$-time algorithm for the $k$-CPM problem.
This is clearly unsatisfactory: it is a simple exercise to design an $\cO(nm)$-time or an $\cO(nk^2)$-time algorithm.
In an effort to tackle this unpleasant situation, we present two much more efficient algorithms: a simple $\cO(nk)$-time algorithm and an $\cO(n+\frac{n}{m}\,{k^4})$-time algorithm.
Our second algorithm applies in an extended way a technique that was developed very recently for $k$-mismatch pattern matching in grammar compressed strings
by Bringmann et al.~\cite{DBLP:conf/soda/BringmannWK19}. We also show that both of our algorithms can be implemented in $\cO(m)$ space.

A preliminary version of this work was published as~\cite{DBLP:conf/fct/Charalampopoulos19}.

\paragraph{\bf Our approach}
We first consider a simple version of the problem (called \textsc{Anchor-Match})
in which we are given a position in $T$ (an \emph{anchor}) which belongs to
potential $k$-mismatch circular occurrences of $P$. A simple $\Oh(k)$-time algorithm is given (after linear-time preprocessing) to compute all relevant occurrences.
By considering separately each position in $T$ as an anchor we obtain an $\Oh(nk)$-time algorithm. 
The concept of an anchor is extended to the so-called \emph{matching pairs}: when we know a pair of positions, one in $P$ and the other in $T$, that are aligned.
Then comes the idea of a \emph{sample} $\sample$, which is a fragment of $P$ of length $\Theta(m/k)$ which supposedly exactly matches a corresponding fragment in $T$.
We choose $\Oh(k)$ samples and work for each of them and for windows of $T$ of size $2m$.
As it is typical in many versions of pattern matching, our solution is split into periodic and non-periodic cases.
If $\sample$ is non-periodic the sample occurs only $\Oh(k)$ times in a window and each occurrence gives a matching pair (and consequently
two possible anchors). Then we perform \textsc{Anchor-Match} for each such anchor. The hard part is the case when $\sample$ is periodic. Here we compute
all exact occurrences of $\sample$ and obtain $\Oh(k)$ groups of occurrences, each one being an arithmetic progression. Now
each group is processed using the approach ``few matches or almost periodicity'' of Bringmann et al.~\cite{DBLP:conf/soda/BringmannWK19}.
In the latter case periodicity is approximate allowing up to $k$ mismatches.
Finally, we are able to decrease the exponent of $k$ by one in the complexity using a marking trick. % similarly to~\cite{DBLP:conf/soda/BringmannWK19} and previous works on approximate pattern matching.

\section{Preliminaries}
%We begin with basic definitions and notation generally following~\cite{AlgorithmsOnStrings}.
Let $S=S[0]S[1]\cdots S[n-1]$ be a \emph{string} of length $|S|=n$ over an integer alphabet $\Sigma$. 
The elements of $\Sigma$ are called \emph{letters}.
For two positions $i$ and $j$ on $S$, we denote by $S[i\dd j]=S[i]\cdots S[j]$ the \emph{fragment} of $S$ that starts at position $i$ and ends at position $j$ (it equals the empty string $\varepsilon$ if $j<i$). 
A \emph{prefix} of $S$ is a fragment that starts at position $0$, i.e.~of the form $S[0\dd j]$,
and a \emph{suffix} is a fragment that ends at position $n-1$, i.e.~of the form $S[i\dd n-1]$.
For an integer $p$, we define the $p$th \emph{power} of $S$, denoted by $S^p$, as the string %SPP: Tomek W, let's avoid dashes when not required.
obtained from concatenating $p$ copies of $S$. $S^\infty$ denotes the string obtained by concatenating infinitely many copies.
If $S$ and $S'$ are two strings of the same length, then by $S =_k S'$ 
we denote the fact that $S$ and $S'$ have at most $k$ mismatches, that is, that the Hamming distance between $S$ and $S'$ does not exceed $k$.

We say that a string $S$ has period $q$ if $S[i] = S[i+q]$ for all $i=0,\ldots,|S|-q-1$.
String $S$ is periodic if it has a period $q$ such that $2q \le |S|$. We denote the smallest period of $S$ by $\per(S)$.
Fine and Wilf's periodicity lemma~\cite{fine1965uniqueness} asserts that if a string of length $n$ has periods $p$ and $q$ and $n \ge p+q-1$, then the string has a period
$\gcd(p,q)$.

For a string $S$, by $\rot_x(S)$ for $0 \le x < |S|$, we denote the string that is obtained from $S$ by moving the prefix of $S$ of length $x$ to its suffix.
We call the string $\rot_x(S)$ (or its representation $x$) a \emph{rotation} of $S$. More formally, we have
\[\rot_x(S)=VU\text{, where }S=UV\text{ and }|U|=x.\]

\subsection{Anatomy of Circular Occurrences}
\setcounter{footnote}{0}

In what follows, we denote by $m$ the length of the pattern $P$ and by $n$ the length of the text $T$. We say that $P$ has a \emph{$k$-mismatch circular occurrence} (in short
\emph{$k$-occurrence}) in $T$ at position $p$ if $T[p \dd p+m-1] =_k \rot_x(P)$ for some rotation $x$.
%Hence, a $k$-occurrence is uniquely described with values $p$ and $x$.
In this case, the position $x$ in the pattern
is called the \emph{split point} of the pattern and $p+(m-x) \bmod m$~\footnote{
The modulo operation is used to handle the trivial rotation with $x=0$.} is called the \emph{anchor} in the text. 
In other words, if $P=UV$ and its rotation $VU$ occurs in $T$, then the first position of $V$
in $P$ is the split point of this occurrence, and the first position of $U$ in $T$ is the anchor of this occurrence (see Fig.~\ref{fig:anchorsplit}).

\begin{figure}[ht]
  \centering
\centerline{\includegraphics[width=6cm]{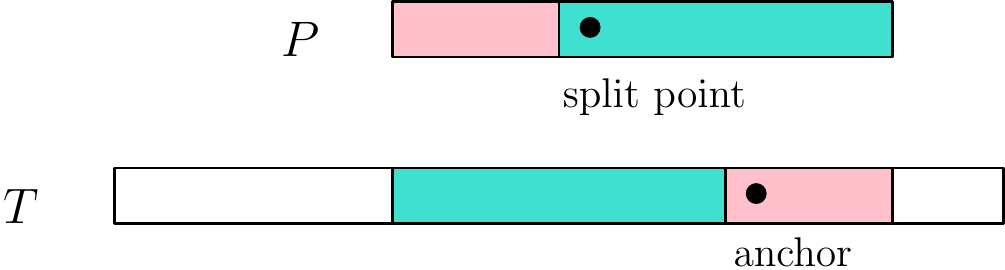}}
  \caption{The split point and the anchor for a $k$-occurrence of $P$ in $T$.
  }\label{fig:anchorsplit}
\end{figure}

The main problem in scope can now be stated as follows.

\defproblem{$k$-CPM Problem}
{
  Text $T$ of length $n$, pattern $P$ of length $m$, and positive integer $k$.
}
{
  All positions of $T$ that contain a $k$-occurrence of $P$.
}

%Let us fix a parameter $m$ which, intuitively, corresponds to the length of the pattern.
For an integer $z$, let us denote $\W_z = [z \dd z+m-1]$ (\emph{window} of size $m$).
Intuitively, this window corresponds to a length-$m$ fragment of the text $T$.
For a $k$-occurrence at position $p$ of $T$ with rotation $x$, 
we introduce a set of pairs of positions in the fragment of the text and the corresponding positions from the original (unrotated) pattern $P$:
\[M(p,x) = \{(i,(i-p+x) \bmod m)\,:\,i \in \W_p\}.\]
The pairs $(i,j)\in M(p,x)$ are called \emph{matching pairs} of an occurrence $p$ with rotation $x$.
In particular, $(p+((m-x)\bmod m),0) \in M(p,x)$. An example is provided in Fig.~\ref{fig:matching_pair}.

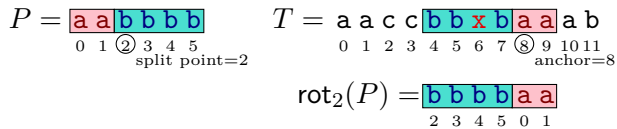
\begin{figure}[ht]
  \centering
  \input{_fig_circular_occ_ex.tex}
  \caption{
    A 1-occurrence of $P\,=\,\mathtt{aabbbb}$ in text $T\,=\,\mathtt{aaccbbxbaaab}$ at position $p=4$ with
    rotation $x=2$;
    $M(4, 2)=\{(4, 2), (5, 3), (6, 4), (7, 5), (8, 0), (9, 1)\}$.
  }\label{fig:matching_pair}
\end{figure}

%------------------------------------------------------------------
\section{An $\cO(nk)$-time Algorithm}\label{sec:nk}

We first introduce an auxiliary problem in which one wants to compute all $k$-occurrences of $P$ in $T$ with a given anchor $\mathbf{a}$.
This problem describes the core computational task in our first solution.

\defproblem{Anchor-Match Problem}
{Text $T$ of length $n$, pattern $P$ of length $m$, positive integer $k$, and position $\mathbf{a}$.
}
{All $k$-occurrences $p$ of $P$ in $T$ with anchor $\mathbf{a}$, represented as a collection of $\Oh(k)$ intervals.
}

\medskip
For a string $X$ let us denote by $X_{(i)}$ and $X^{(i)}$ its fragments of length $m$ starting at 
position $i$ and ending at position $i$, respectively.
Moreover, for a binary string $X$, by $||X||$ we denote the arithmetic sum of characters in $X$.
We define the following auxiliary problem.

\defproblem{Light-Fragments Problem}
{Positive integers $m$, $k$ and a string $V$ of length $n$ over alphabet $\{0,1\}$ containing $\Oh(k)$ non-zero characters.
The string $V$ is specified by its positions with non-zero characters (sorted increasingly).
}
{The set $A\;=\; \{\,i\ :\  ||V_{(i)}||\,\le\, k\,\}$
represented as a collection of $\Oh(k)$ intervals.
} 

\begin{lemma}\label{Nov6}
The \textsc{Light-Fragments} problem can be solved in $\Oh(k)$ time. 
\end{lemma}
\begin{proof}
  Let $I$ be the set of positions of $V$ with non-zero characters; $|I| = \Oh(k)$ by definition.
  We define a piecewise constant function $f$ on integers such that $f(x) = | I \cap [0,x] |$.
  Let $g(x) = f(x+m-1)-f(x-1)$.
  Then $g(x) = | I \cap [x,x+m-1] |$.
  Function $f$ has $\Oh(k)$ pieces, so $g$ has $\Oh(k)$ pieces as well and both can be computed in $\Oh(k)$ time since $I$ is sorted.
  In the end we report the pieces where $g$ has value at most $k$.
\end{proof}

Let us recall a standard algorithmic tool for preprocessing text $T$.
We denote the length of the longest common prefix (resp.~suffix) of two strings $U$ and $V$ by $\lcp(U,V)$ (resp.~$\lcsuf(U, V)$).
There is an $\Oh(n)$-sized data structure answering such queries over suffixes (resp.~prefixes) of $T$ in $\Oh(1)$ time after $\Oh(n)$-time preprocessing.
It consists of the suffix array of $T$ and a data structure for answering range minimum queries; see~\cite{AlgorithmsOnStrings}.
Using the kangaroo method~\cite{DBLP:journals/tcs/LandauV86,DBLP:journals/tcs/GalilG87}, 
these queries can handle mismatches; after an $\cO(n)$-time preprocessing of $T$,
longest common prefix (resp.~suffix) queries with up to $k$ mismatches can be answered in $\cO(k)$ time.

\begin{lemma}\label{lem:alg-anchor-match-problem}
  After $\cO(n)$-time preprocessing, the answer to \textsc{Anchor-Match} problem, represented as a union of $\Oh(k)$ intervals, can be computed in $\Oh(k)$ time.
\end{lemma}
\begin{proof}
  In the preprocessing we prepare a data structure for $\lcp$  queries in $P \# T$, for a special
  character $\#$ that occurs neither in $P$ nor in $T$.
 
  Consider now a query for an anchor $\mathbf{a}$ over $T$. Let $L=T^{(\mathbf{a}-1)}$ and $R=T_{(\mathbf{a})}$.
  We define a binary string $L'$ such that $L'[i]=1$ if and only if $L[i]\ne P[i]$. We define $R'$ analogously. 
  Let $L''$ be the longest suffix of $L'$ such that $||L''||\le k$. Let $R''$ be the longest prefix of $R'$ such that $||R''|| \le k$.
 
  Using the kangaroo method~\cite{DBLP:journals/tcs/LandauV86,DBLP:journals/tcs/GalilG87}, the strings $L''$, $R''$ can be constructed in $\Oh(k)$ time.
  The \textsc{Anchor-Match} problem now reduces to the \textsc{Light-Fragments} problem for string $V=L''R''$. 
\end{proof}

For an illustration of~\cref{lem:alg-anchor-match-problem} inspect Fig.~\ref{fig:light_fragments}.

\begin{figure}[ht]
  \centering
  \input{_fig_light_fragments_ex.tex}
  \caption{An illustration of the setting in~\cref{lem:alg-anchor-match-problem} with $P=(\texttt{abaababa})^2$, the text as in the figure, anchor $\mathbf{a}=16$ and $k=3$. The string $V$, used in the proof of the lemma, shows the positions of the at most $k$ mismatches to the left and to the right of the anchor. The output consists of the three intervals $[1,3]$, $[7,8]$ and $[13,14]$ shown in orange. For example, 7 is an occurrence since the fragment of $V$ of length $|P|$ starting at this position (represented by a rectangle) contains at most $k$ ones.
  }\label{fig:light_fragments}
\end{figure}
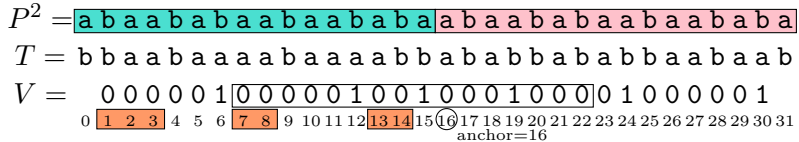
  
%The preprocessing are LCE queries.

%For each position $i$ in $t$ and $q=0,\ldots,k$, let us compute 
%$LCP_q(t[i\dd ],p)$ and $LCS_q(t[\dd i-1], p)$ 
%(lcp - Longest common prefix between text and pattern with $q$ mismatches,
% lcs - Longest common suffix between text and suffix of pattern with $q$ mismatches).

%We need to locate pairs $(i, q)$ such that $LCP_q[i] + LCS_q[i-1] \ge m$.
%Each such pair define a range of occurrences.

%\begin{proof}
 
  %Let $i_1,\ldots,i_k$ be the mismatches to the left and $j_1,\ldots,j_k$ be the mismatches to the right.
  %Let $k_1$ be the number of mismatches in the left part.
  %This means that the occurrence must fit within $(i_{k_1+1},j_{k-k_1+1})$.
  %This gives a range of starting positions of an occurrence.
%\end{proof}

%In the end, we need to sort the $\cO(nk)$ ranges and compute their union.
%This can be done in $\cO(n)$ space.
%To this end, we process the ranges in batches of $n$.
%For each batch we sort the ranges using bucket sort, compute the union in a left-to-right traversal, and store the result as a 0/1 vector of size $n$.

\begin{proposition}\label{prop:nk}
 $k$-CPM can be solved in $\cO(nk)$ time and $\cO(n)$ space.
\end{proposition}
\begin{proof}
We invoke the algorithm of Lemma~\ref{lem:alg-anchor-match-problem} for all $\mathbf{a} \in [0 \dd n-1]$ and obtain $\cO(nk)$ intervals of $k$-occurrences of $P$ in $T$.
Instead of storing all the intervals, we count how many intervals start and end at each position of $T$.
We can then compute the union of the intervals by processing these counts from left to right.
\end{proof}

\section{Algorithmic tools}
In this section we introduce further algorithmic tools to get to our second solution. 

\subsection{Internal Queries in a Text}

Let $T$ be a string of length $n$ called text.
An Internal Pattern Matching (IPM) query, for two given fragments $F$ and $G$ of the text, such that $|G| \le 2|F|$,
computes the set of all occurrences of $F$ in $G$.
If there are more than two occurrences, they form an arithmetic sequence with difference $\per(F)$.
A data structure for IPM queries in $T$ can be constructed in $\Oh(n)$ time and answers queries in $\Oh(1)$ time
(see~\cite{DBLP:conf/soda/KociumakaRRW15} and~\cite[Theorem 1.1.4]{Kocium}).
It can be used to compute all occurrences of a given fragment $F$ of length $p$ in $T$, expressed as a union of $\Oh(n/p)$ pairwise disjoint
arithmetic sequences with difference $\per(F)$, in $\Oh(n/p)$ time.

\subsection{Simple Geometry of Arithmetic Sequences of Intervals}\label{subsec:chains}

%Before we proceed with showing how to efficiently handle periodic samples, we present algorithms that will be used in the proofs for handling regular sets of intervals.
We next present algorithms that will be used in subsequent proofs for handling regular sets of intervals.

For an interval $I$ and an integer $r$, let $I\oplus r\;=\; \{\,i+r\::\; i\in I\,\}$.
We define \[\Chain_q(I,a)\;=\; I\,\cup\, (I\oplus q)\,\cup\, (I\oplus 2q)\,\cup \dots\cup\, (I\oplus aq).\]
This set is further called an \emph{interval chain} (with difference $q$).
Note that it can be represented in $\Oh(1)$ space using four integers: $a$, $q$, and the endpoints of $I$. 
An illustration of an interval chain representing the output for a problem defined in the next subsection can be found in Fig.~\ref{fig:int_chain}.

For a given value of $q$, let us fit the integers from $[1\dd n]$ into the cells of a grid of width $q$ so that the first row
consists of numbers 1 through $q$, the second of numbers $q+1$ to $2q$, etc.
Let us call this grid $\G_q$.
A chain $\Chain_q$ can be conveniently represented in the grid $\G_q$ using the following lemma from~\cite{DBLP:conf/lata/KociumakaRRSWZ19}.

\begin{lemma}[\cite{DBLP:conf/lata/KociumakaRRSWZ19}]\label{lem:chain_rect}
  The set $\Chain_q (I,a)$ is a union of $\Oh(1)$ orthogonal rectangles in $\G_q$.
  The coordinates of the rectangles can be computed in $\Oh(1)$ time.
\end{lemma}

Lemma~\ref{lem:union_chains} can be used to compute a union of interval chains.

\begin{lemma}\label{lem:union_chains}
  Given $c$ interval chains, all of which have difference $q$ and are subsets of $[0 \dd n]$, the union of these chains, expressed as a subset of $[0 \dd n]$, can be computed in $\Oh(n+c)$ time.
\end{lemma}

\begin{proof}
  By Lemma~\ref{lem:chain_rect}, the problem reduces to computing the union of $\Oh(c)$ rectangles on a grid of total size $n$.
  Let $t$ be a 2D array of the same shape as $\G_q$, initially set to zeroes.
  For a rectangle with opposite corners $(x_1,y_1)$ and $(x_2,y_2)$, with $x_1 \le x_2$ and $y_1 \le y_2$, we increment $t[x_1,y_1]$,
  decrement $t[x_2+1,y_1]$ and $t[x_1,y_2+1]$, and increment $t[x_2+1,y_2+1]$ (provided that the respective cells are within the array).
  This takes $\Oh(c)$ time.
  We then compute prefix sums of $t$, which are defined as
  \[t'[x,y]=\sum_{i=1}^x \sum_{j=1}^y t[i,j].\]
  Such values can be computed in time proportional to the size of the grid, i.e., in $\Oh(n)$ time.
  Finally, we note that $(x,y)$ is contained in $t'[x,y]$ rectangles, concluding the proof.
\end{proof}

\begin{remark}
The proof of~\cref{lem:union_chains} is essentially based on an idea that was used, for example, 
for reducing the decision version of range stabbing queries in 2D to weighted range counting queries in 2D (cf.~\cite{DBLP:journals/siamcomp/Patrascu11}).
\end{remark}

We will also use the following auxiliary lemma.

\begin{lemma}\label{lem:3chains}
  Let $X$ and $Z$ be intervals and $q$ be a positive integer.
  The set \[Z':=\{z \in Z\,:\, \exists_{x\in X}\, z \equiv x \pmod q\},\] represented as a disjoint sum of at most three interval chains, each with difference $q$,  can be computed in $\Oh(1)$ time.
\end{lemma}
\begin{proof}
 If $|X|\ge q$, then $Z'=Z$ is an interval and thus an interval chain.
 % In this case we divide $Z$ into two parts -- the first part consists of
 %first $q\cdot \lfloor\frac{|Z|}{q}\rfloor$ elements, and the second part of the remaining ones:\newline
 %$\Chain_q([q]\oplus z_0,\lfloor\frac{|Z|-1}{q}\rfloor)\cup\Chain_q([|Z| \bmod q]\oplus z_0,0)$, where $z_0$ is first element of $Z$, and $[n]={0,1,2,...,n-1}$.
 If $|X|<q$, then $Z'$ can be divided into disjoint intervals of length smaller than or equal to $|X|$. The intervals from the second until the penultimate one (if any such exist),
 have length $|X|$.
 Hence, they can be represented as a single chain, as the first element of each such interval is equal $\bmod \; q$ to the first element of $X$.
 The two remaining intervals can be treated as chains as well.
 %-- let $z_2=min\{z\in Z:z=x_0\bmod q\}$, and $z_3=max\{z\in Z:z=x_1\bmod q\}$
 %(where $x_0,x_1$ are the first and last elements of $X$). If one of those does not exist, then
\end{proof}

\subsection{The Aligned-Light-Sum Problem}
We define the following abstract problem that resembles the \textsc{Light-Fragments} problem from Section~\ref{sec:nk}.

\defproblem{Aligned-Light-Sum Problem}
{Positive integers $m$, $k$, $q$ and strings $U,V$ over alphabet $\{0,1\}$, each containing $\Oh(k)$ non-zero characters.
The strings are specified by their positions with non-zero characters.
}
{The set $A\;=\; \{\,i\ :\  (\exists\, j) \; ||U_{(i)}||\,+\,||V_{(j)}||\,\le\, k\ \land\ j \equiv i \pmod{q}\}$.
}

\begin{figure}[thpb]
  \centering
  \input{_fig_int_chains.tex}
  \caption{An instance of the \textsc{Aligned-Light-Sum} problem with $m=15$, $k=2$, $q=4$, $U=\texttt{0} \texttt{1} \texttt{0}^{14}$ and $V= \texttt{0}^{14} \texttt{1} \texttt{0}^{14}$. 
	The output is the interval chain $\Chain_4([0,1],3)$ shown in orange.}\label{fig:int_chain}
\end{figure}

\begin{lemma}\label{Oct31}
The \textsc{Aligned-Light-Sum} problem can be solved in $\Oh(k^2)$ time with the output represented as a collection of $\Oh(k^2)$ interval chains, each with difference $q$.
\end{lemma}
\begin{proof}
  Let $I$ and $I'$ be the positions with non-zero characters in $U'$ and $V'$, respectively.
	We partition the set $\{0,\ldots,|U'|-m\}$ into intervals such that for all $j$ in an interval, the set $\W_j \cap I$ is the same.
  For this, we use a sliding window approach.
  We generate events corresponding to $x$ and $x-m+1$ for all $x \in I$ and sort them.
  When $j$ crosses an event, the set $\W_j \cap I$ changes.
  Thus we obtain a partition of $\{0,\ldots,|U'|-m\}$ into intervals $Z_1,\ldots,Z_{n_1}$.
  We obtain a similar partition of $\{0,\ldots,|V'|-m\}$ into intervals $Z'_1,\ldots,Z'_{n_2}$.
  We have $n_1,n_2 = \Oh(k)$.

  Let us now fix $Z_j$ and $Z'_{j'}$. % (see also Fig.~\ref{fig:abstract2}).
  First we check if the condition on the total number of non-zero characters is satisfied for any $z \in Z_j$ and $z' \in Z'_{j'}$.
  If so, we compute the set $Z'_{j'} \bmod q =\{z' \bmod q\,:\,z' \in Z'_{j'}\}$.
  It is a single circular interval and can be computed in constant time.
  The required result is 
  \[\{z \in Z_j\,:\, z \bmod q \in X\}.\]

\smallskip\noindent 
	By Lemma~\ref{lem:3chains}, this set can be represented as a union of 
	three chains, each with difference $q$ and, as such, can be computed in $\Oh(1)$ time.
  The conclusion follows.
\end{proof}
%------------------------------------------------------------------

\section{An $\cO(n+{k^5})$-time Algorithm for Short Texts}

In this section we proceed by assuming that $m \le n\le 2m$ and aim at an $\cO(n+k^5)$-time algorithm. In the next sections we 
remove this assumption and reduce the exponent of $k$ to 4.

A \emph{(deterministic) sample} is a short fragment $\sample$ of the pattern $P$.
An occurrence in the text without any mismatch is called {\it exact}.
%We consider the set of starting positions of exact occurrences of the sample in $T$.
We introduce a problem of \textsc{Sample-Match} that consists in finding all $k$-occurrences of $P$ in $T$ such that $\sample$
matches exactly a fragment of length $|\sample|$ in $T$.

We split the pattern $P$ into $2k+3$ fragments of length $\floor{\frac{m}{2k+3}}$ or $\ceil{\frac{m}{2k+3}}$ each.
In any $k$-occurrence of $P$ in $T$ at least $k+2$ of those fragments will occur exactly in $T$ (up to $k$ fragments may occur with one mismatch and at most one fragment will contain the split point).

\begin{remark}
We force at least $k+2$ fragments (instead of just one) to match exactly for two reasons: (1) in order to have more than a half of them match exactly, which will guarantee that the interval chains that are obtained from applications of the \textsc{Aligned-Light-Match} problem have the same difference and thus can be unioned using Lemma~\ref{lem:union_chains} (see the proof of Proposition~\ref{prop:nk5}); and (2) for the marking trick in the next section.
\end{remark}

Let us henceforth fix a sample $\sample$ as one of these fragments, let $p_{\sample}$ be its starting position in $P$, and let $m_{\sample} = |\sample|$.
We assume that the split point $x$ in $P$ is to the right of $\sample$, i.e., that $x \ge p_{\sample}+m_{\sample}$.
The opposite case---that $x < p_{\sample}$---can be handled analogously.

\subsection{Matching Non-Periodic Samples}
Let us assume that $\sample$ is non-periodic. 
Further let $j$ denote the starting position of $\sample$ in $P$ and $i$ denote a starting position of an occurrence of $\sample$ in $T$.
We introduce a problem in which, intuitively, we compute all $k$-occurrences of $P$ in $T$ which align $T[i]$ with $P[j]$.

\defproblem{Pair-Match Problem}
{
  Text $T$ of length $n$, pattern $P$ of length $m$, positive integer $k$, and two integers $i \in [0\dd n-1]$ and $j \in [0\dd m-1]$.
}
{
  The set $A(i,j)$ of all positions in $T$ where we have a $k$-mismatch occurrence of $\rot_x(P)$ for some $x$ such that $(i,j)$ is a matching pair.
}

We show how to solve the \textsc{Sample-Match} problem for a non-periodic sample $\sample$ in $\Oh(k^2)$ time in two steps.
First, we show how to solve the \textsc{Pair-Match} problem for a given $(i,j)$ in $\Oh(k)$ time.
We then apply this solution for each occurrence $i$ of $\sample$ in $T$ by noticing that we can have $\Oh(k)$ such occurrences in total.

% \begin{observation}
%   \label{obs:anchors}
%   There are at most two possible anchors of occurrences
%   with a given matching pair of positions $(i,j)$, specifically (see also Fig.~\ref{fig:2anchors}):
%   \begin{itemize}
%     \item anchor $i-j$ (if $i-j \le 0$), and
%     \item anchor $i+|P|-j$ (if $i+|P|-j < |T|$).
%   \end{itemize}
%   However, there are many possible rotations $x$ for each of the two anchors. 
% \end{observation}

\begin{figure}[ht]

    \centering
    \vspace*{-1cm}
    \includegraphics[width=.45\textwidth]{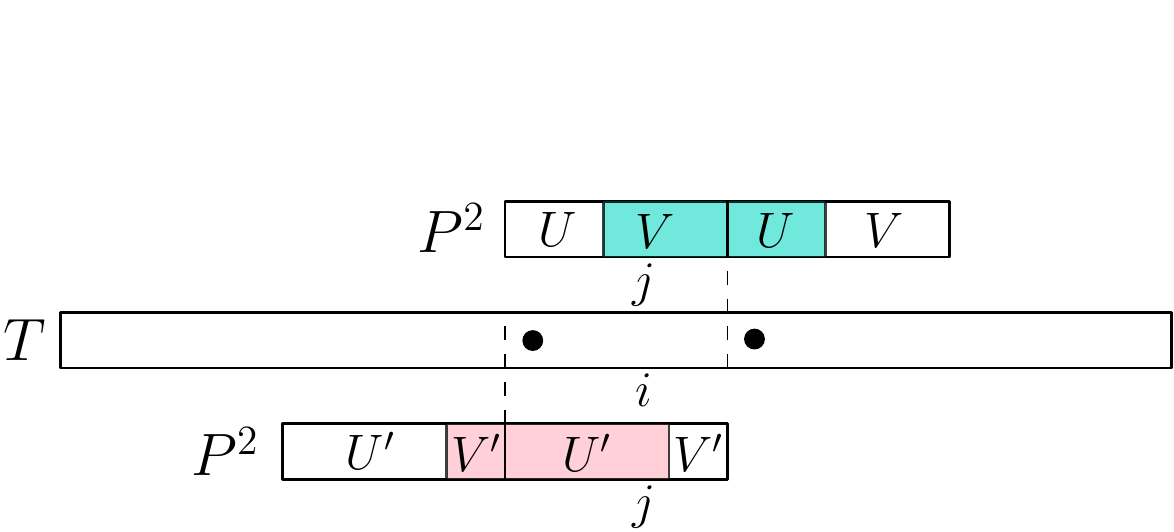}
	\caption{The two possible anchors for the matching pair of positions $(i,j)$ are shown as bullet points. 
	A possible $k$-occurrence of $P$ in $T$ corresponding to the left (resp.~right) anchor is shown below $T$ (above $T$, resp.). Note that $P^2=PP$.}\label{fig:2anchors}
\end{figure}

\begin{lemma}\label{lem:mpp}
After $\cO(n)$-time preprocessing, the \textsc{Pair-Match} problem can be solved in $\Oh(k)$ time, where the output is represented as a union of $\Oh(k)$ intervals.
\end{lemma}
\begin{proof}
  Recall that the \textsc{Anchor-Match} problem returns all $k$-occurrences of $P$ in $T$ with a given anchor.
  The \textsc{Pair-Match} problem can be essentially reduced to the \textsc{Anchor-Match} problem,
  since for a given matching pair of characters in $P$ and $T$, there
  are at most two ways of choosing the anchor depending on the relation between $j$ and a split point:
  %these are $i-j$ (if $i-j \ge 0$) and $i+|P|-j$ (if $i+|P|-j < |T|$).
  these are $i-j$ and $i+|P|-j$ (see~\cref{fig:2anchors}).
  Clearly, we choose $i-j$ as an anchor only if $i-j \ge 0$ and $i+|P|-j$ only if $i+|P|-j < |T|$.
  We then have to take the intersection of the answer with $[i-m+1\dd i]$ to ensure that the $k$-occurrence contains position $i$.
\end{proof}

\begin{lemma}\label{lem:nonper}
After $\Oh(n)$-time preprocessing, the \textsc{Sample-Match} problem for a non-periodic sample $\sample$ can be solved in $\Oh(k^2)$ time and outputs a union of $\Oh(k^2)$ intervals
of occurrences.
\end{lemma}
\begin{proof}
Since $\sample$ is non-periodic, it has $\cO(k)$ occurrences in $T$, which can be computed in $\cO(k)$ time after an $\cO(n)$-time preprocessing
using IPM queries~\cite{DBLP:conf/soda/KociumakaRRW15,Kocium} in $P\#T$.
Let $j$ be the starting position of $\sample$ in $P$ and $i$ be a starting position of an occurrence of $\sample$ in $T$.
For each of the $\cO(k)$ such pairs $(i,j)$, the computation reduces to the \textsc{Pair-Match} problem for $i$ and $j$. The statement follows by Lemma~\ref{lem:mpp}.
\end{proof}

\subsection{Matching Periodic Samples}

Let us assume that $\sample$ is periodic, i.e., it has a period $q$ with $2q \leq |\sample|$.
A fragment of a string $S$ containing an inclusion-maximal arithmetic sequence of occurrences of $\sample$ in $S$ with difference $q$ is called here an \emph{$\sample$-run}.
If $\sample$ matches a fragment in the text, then the match belongs to an $\sample$-run. For example, the underlined 
fragment of $S=\texttt{bb\underline{ababab}aa}$ is an $\sample$-run for $\sample=\texttt{abab}$.

\begin{lemma}\label{lem:P1_runs}
  If $\sample$ is periodic, the number of $\sample$-runs in the text is $\Oh(k)$ and they can all be computed in $\Oh(k)$ time
  after $\Oh(n)$-time preprocessing.
\end{lemma}
\begin{proof}
  We construct the data structure for IPM queries on $P \# T$.
  This allows us to compute the set of all occurrences of $\sample$ in $T$ as a collection of $\Oh(k)$ arithmetic sequences with difference $\per(\sample)$.
  We then check for every two consecutive sequences if they can be joined together.
  This takes $\Oh(k)$ time and results in $\Oh(k)$ $\sample$-runs.
\end{proof}

\noindent
For two equal-length strings $S$ and $S'$, we denote the set of their \emph{mismatches} by
\[\Mis(S,S') = \{i=0,\ldots,|S|-1\,:\,S[i] \ne S'[i]\}.\]
%Let $Q=S[i \dd j]$. 
We say that position $a$ in $S$ is a \emph{misperiod} with respect to the fragment $S[i \dd j]$ if $S[a] \ne S[b]$ where $b$ is the unique position such that $b \in [i \dd j]$ and $(j-i+1) \mid (b-a)$.
%We first present Fig.~\ref{fig:bring} to convey the intuition behind the combinatorial property presented by Bringmann et al.~\cite{DBLP:conf/soda/BringmannWK19}, and then move on to formalise it.

We define the set $\LeftMis_k(S,i,j)$ as the set of $k$ maximal misperiods that are smaller than $i$ and $\RightMis_k(S,i,j)$ as the set of $k$ minimal misperiods that are greater than $j$.
Each of the sets can have less than $k$ elements if the corresponding misperiods do not exist.
We further define
\[\Misp_k(S,i,j)=\LeftMis_k(S,i,j) \cup \RightMis_k(S,i,j)\]
and $\Misp(S,i,j) = \bigcup_{k=0}^\infty \Misp_k(S,i,j)$.
%We also denote
%$\LeftMis(Q,S) = \bigcup_{k=1}^\infty \LeftMis_k(Q,S)$ and
%$\RightMis(Q,S) = \bigcup_{k=1}^\infty \RightMis_k(Q,S)$

The following lemma captures a combinatorial property behind the new technique of Bringmann et al.~\cite{DBLP:conf/soda/BringmannWK19}. The intuition is shown in Fig.~\ref{fig:bring}.

\begin{lemma}\label{lem:JR}
  Assume that $S =_k S'$ and that $S[i \dd j] = S'[i \dd j]$.
  Let
    \[I=\Misp_{k+1}(S,i,j) \text{ and } I'=\Misp_{k+1}(S',i,j).\]
  If $I \cap I' = \emptyset$, then $\Mis(S,S')=I \cup I'$, $I=\Misp(S,i,j)$, and $I'=\Misp(S',i,j)$.
\end{lemma}

\begin{proof}
Let $J=\Misp(S,i,j)$ and $J'=\Misp(S',i,j)$.
We first observe that $I \cup I' \subseteq \Mis(S,S')$ since $I \cap I' = \emptyset$.
Then, $S =_k S'$ implies that $|\Mis(S,S')| \leq k$ and hence $|I|\leq k$ and $|I'| \leq k$, which in turn implies that $I=J$ and $I'=J'$.
The observation that $\Mis(S,S') \subseteq J \cup J'$ concludes the proof.
\begin{comment}
  Without loss of generality, we may assume that $i=0$ (otherwise, we consider the prefixes ending at $j$ and the suffixes starting at $i$).

  Now, we proceed by induction on $k$.
  If $k=0$, then $I=I'$, so $I\cap I'=\emptyset$ implies $I\cup I' = \emptyset = \Mis(S,S')$ as claimed.
  For $k\ge 0$, consider the longest prefix $\bar{S}$ of $S$ and $\bar{S}'$ of $S'$ such that $\bar{S} =_{k-1} \bar{S}'$.
  Observe that $\bar{I}=\RightMis_{k}(\bar{S},0,j) \subseteq I$ and $\bar{I'}=\RightMis_{k}(\bar{S}',0,j)\subseteq I'$.
  Hence, $I\cap I' = \emptyset$ implies $\bar{I} \cap \bar{I'} =\emptyset$, so we can use the inductive assumption
  to conclude that $\Mis(\bar{S},\bar{S}') = \bar{I}\cup \bar{I'}$.
  In particular, $\bar{I}+\bar{I'} \le k-1$. If $S=\bar{S}$, this yields $I=\bar{I}$ and $I'=\bar{I'}$, so the claim holds.
  Otherwise, $\Mis(S,S')=\{|\bar{S}|\}\cup \Mis(\bar{S},\bar{S}')$.
  Moreover, $|\bar{S}|\in \Misp_{k}(S,0,j)$ since this position is a misperiod in $S$ or $S'$.
  Finally, we need to prove that there is no larger misperiod in $S$ or $S'$.
  The leftmost such position would be a misperiod in both $S$ and $S'$ (because $|\bar{S}|$ is the rightmost mismatch),
  so it would belong to $\RightMis_{k+1}(S,0,j)\cap \RightMis_{k+1}(S',0,j)=I\cap I'$, contradicting the assumption that $I\cap I' = \emptyset$.
  Consequently, $I\cup I' = \bar{I}\cup \bar{I'}\cup |\bar{S}| = \Mis(S,S')$, as claimed.
  Finally, since $|I \cup I'| \leq k$, the last statement of the lemma follows.
\end{comment}
\end{proof}

\begin{figure}[htpb]
  \centering
  \input{_fig_bring.tex}
  \caption{Let $S$, $S'$, and $X$ be equal-length strings such that $X$ is a factor of $Q^\infty$ and $S[i \dd j] = S'[i \dd j] = X[i \dd j]=Q$. The asterisks in $S$ denote the positions in $\Mis(S,X)$, or equivalently, the misperiods with respect to $S[i \dd j]$. Similarly for $S'$. One can observe that $\Mis(S,X) \cap \Mis(S',X)= \emptyset$ and that $\Mis(S,X) \cup \Mis(S',X) = \Mis(S,S')$.
  }\label{fig:bring}
\end{figure}
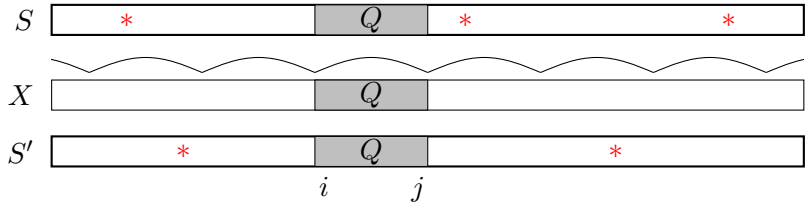

\begin{comment}
\begin{theorem}[Bringmann et al.~\cite{DBLP:conf/soda/BringmannWK19}]
Given strings $p$ of length $m$ and $t$ of length at most $2m$, at least one of the following holds:
\begin{itemize}
\item The number of $k$-matches of $p$ in $t$ is at most $\Oh(k^2)$.
\item Let $t'$ denote the shortest substring of $t$ such that any $k$-match of $p$ in $t$ is also a $k$-match in $t'$.
Then there is a substring $p_{\sample}$ of $p$, with $|p_{\sample}| \le m/k$, such that $d_H(p,(p_{\sample})^*[1,m]) \le \Oh(k)$ and $d_H(t',(p_{\sample})^*[1,|t'|])\le \Oh(k)$.
Moreover, any $k$-match of $p$ in $t'$ starts at a position of the form $1 +i\cdot |p_{\sample}|$ with $0\le i\le (|t'|-|p|)/|p_{\sample}|$
(but not every starting position $1 +i\cdot |p_{\sample}|$ necessarily yields a $k$-match).
\end{itemize}
\end{theorem}
\end{comment}

A string $S$ is \emph{$k$-periodic w.r.t.\ an occurrence $i$ of $Q$} if $|\Misp(S,i,i+|Q|-1)| \le k$.
In this case, $|Q|$ is called the \emph{$k$-period}.
In particular, in the conclusion of the above lemma $S$ and $S'$ are $|I|$-periodic and $|I'|$-periodic, respectively, w.r.t. $Q=S[i \dd j] = S'[i \dd j]$.
This notion forms the basis of the following auxiliary problem in which we search for $k$-occurrences in which the rotation of the pattern and
the fragment of the text are $k$-periodic for the same period $Q$.

Let $U$ and $V$ be two strings and $J$ and $J'$ be sets containing positions in $U$ and $V$, respectively.
We say that length-$m$ fragments $U[p \dd p+m-1]$ and $V[x \dd x+m-1]$ are \emph{$(J,J')$-disjoint}
if the sets $(\W_p \cap J) \ominus p$ and $(\W_x \cap J') \ominus x$ are disjoint.
For example, if $J=\{2,4,11,15,16,17\}$, $J'=\{5,6,15,18,19\}$, and $m=12$,
then $U[3 \dd 14]$ and $V[6 \dd 17]$ are $(J,J')$-disjoint for:

\setlength{\fboxrule}{0.1pt} 
\begin{center}
$U=$\ \texttt{\ \ \ \,ab$\bullet$\,\framebox{a$\bullet$b\,abc\,ab$\bullet$\,abc}\,$\bullet${$\bullet$}$\bullet$}\\
$V=$\ \texttt{abc\,ab$\bullet$\,\framebox{{$\bullet$}bc\,abc\,abc\,{$\bullet$}bc}\,$\bullet${$\bullet$}c}
\end{center}

Let us introduce an auxiliary problem that is obtained in the case that is shown in the conclusion of the above lemma
(i.e., misperiods in the rotation of $P$ and the corresponding fragment of $T$ are not aligned); see also Fig.~\ref{fig:primitive}.

\defproblem{Periodic-Periodic-Match Problem}{
A string $U$ which is $2k$-periodic w.r.t.\ an exact occurrence $i$
of a length-$q$ string $Q$ and 
a string $V$ which is $2k$-periodic w.r.t.\ an exact occurrence $i'$ of the same string $Q$
such that $m \le |U|,|V| \le 2m$ and

\smallskip
\centerline{
 $J=\Misp(U,i,i+q-1),\  \ J'=\Misp(V,i',i'+q-1)$.}

\smallskip
(The strings $U$ and $V$ are not stored explicitly.)
}{
The set of positions $p$ in $U$ for which there exists a
$(J,J')$-disjoint $k$-occurrence $U[p \dd p+m-1]$ of $V[x \dd x+m-1]$
% position $x$ in $V$ 
for $x$ such that
%\begin{enumerate}
%  \item $U[p \dd p+m-1] =_k V[x \dd x+m-1]$
%  \item 

\smallskip
\centerline{
$i-p \equiv i'-x \pmod {q}$.}
%  \item The sets $\W_p \cap J$ and $(\W_x \cap J') \oplus (p-x)$ are disjoint
%\end{enumerate}
}

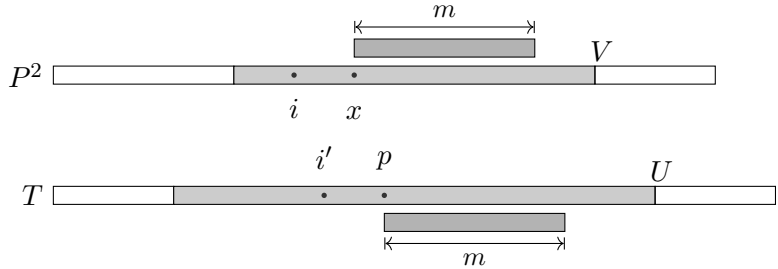
\begin{figure}[htpb]
  \centering
  %\centerline{\includegraphics[width=8cm]{primitive.pdf}}
  \input{_fig_primitive.tex}
  \caption{In the case of a periodic sample, we have two $2k$-periodic fragments $U$ and $V$.
	We find all $p$ in $U$ such that for some position $x$ in $V$, the fragments of length $m$ starting at positions $p$ and $x$ are at Hamming distance at most $k$.
  If $q$ is a period, then $i-p \equiv i'-x \pmod q$ due to synchronization of periodicities.}\label{fig:primitive}
\end{figure}

\smallskip
\noindent
Intuitively, the modulo condition on the output of the \textsc{Periodic-Periodic-Match} problem corresponds to the fact that the approximate periodicity is aligned.

In the \textsc{Periodic-Periodic-Match} problem we search for $k$-occurrences in 
which none of the misperiods in $J$ and $J'$ are aligned. In this case each of 
the misperiods accounts for one mismatch in the $k$-occurrence. 
In the lemma below we reduce the \textsc{Periodic-Periodic-Match} problem to the \textsc{Aligned-Light-Sum}
problem, in which we only require that the total number of misperiods in an 
occurrence is at most $k$. This way all the $(J,J')$-disjoint $k$-occurrences 
can be found. Also additional occurrences where two misperiods are aligned can be 
reported, but they are still valid $k$-occurrences (actually, $k'$-occurrences 
for some $k' < k$). % duplicates of the ones found by \textsc{Pairs-Match} routine in Run-Sample-Matching.

\begin{lemma}\label{lem:ppm}
We can compute in $\Oh(k^2)$ time a set of $k$-occurrences of $P$ in $T$ represented as $\Oh(k^2)$ interval chains, each with difference $q$, that is a superset of the solution to the \textsc{Periodic-Periodic-Match} problem.
%(Some of the computed $k$-occurrences may not satisfy the third condition of the problem.)
\end{lemma}
\begin{proof}
In the \textsc{Periodic-Periodic-Match} problem the modulo condition forces the exact occurrences of the approximate period to match.
Hence, it guarantees that all the positions except the misperiod positions match.
Now the \textsc{Aligned-Light-Sum} problem highlights these positions inside the string.
%(see also Fig.~\ref{fig:abstract1}).

\begin{claim} 
\textsc{Periodic-Periodic-Match} can be reduced in $\Oh(k)$ time to the \textsc{Aligned-Light-Sum} problem
so that we obtain a superset of the desired result.
The potential extra positions do not satisfy only the $(J,J')$-disjointness condition.
\end{claim}
\begin{proof}
Let the parameters $m$, $k$ and $q$ remain unchanged.
We create strings $U'$ and $V'$ of length $|U|$ and $|V|$, respectively, with positions with non-zero characters in the sets $I$ and $I'=J'$, respectively.
Then we prepend $U'$ with $z=(i'-i) \bmod q$ zeros.
Let $A$ be the solution to the \textsc{Aligned-Light-Sum} problem for $U'$ and $V'$.
Then $(A \ominus z) \cap \mathbb{Z}_{\ge 0}$ is a superset of the solution to \textsc{Periodic-Periodic-Match};
the elements of the set that correspond to matches where non-zero elements of the strings $U'$, $V'$ were aligned do not satisfy the disjointness condition.
\end{proof}
Now the thesis follows from Lemma~\ref{Oct31}.
\end{proof}

%\noindent
%In the solution we do not check if the sets $(\W_p \cap J) \ominus p$ and $(\W_x \cap J') \ominus x$ are disjoint.
%However, a $k'$-occurrence is found for some $k'<k$ otherwise.

Let us further define 

\[\textsc{Pairs-Match}(T,I,P,J)=\bigcup_{i \in I, j \in J} \textsc{Pair-Match}(T,i,P,j).\]

Let $A$ be a set of positions in a string $S$ and $m$ be a positive integer.
We then denote $A \bmod m = \{a \bmod m\,:\,a \in A\}$
and by $\frag_A(S)$ we denote the fragment $S[\min(A) \dd \max(A)]$.
We provide pseudocode for an algorithm that computes all $k$-occurrences of $P$ such 
that $\sample$ matches a fragment of a given $\sample$-run (see~\cref{alg:periodic}); inspect also Fig.~\ref{fig:app:periodic}.

\begin{algorithm}
\vspace*{0.3cm}
\caption{Run-Sample-Matching}\label{alg:periodic}
\KwData{A periodic fragment $\sample$ of pattern $P$, an $\sample$-run $R$ in text $T$, $q=\per(\sample)$, and $k$.}
\KwResult{A compact representation of $k$-occurrences of $P$ in $T$ including all $k$-occurrences\\ where $\sample$ in $P$ matches a fragment of $R$ in $T$.}

\setstretch{1.2}
Let $R=T[s \dd s+|R|-1]$\; %and $Q=P[p_{\sample} \dd p_{\sample}+q-1]$\;
$J:=\Misp_{k+1}(T,s,s+q-1)$; \{ $\Oh(k)$\,time \}\\
$J':=\Misp_{k+1}(P^2,m+p_{\sample},m+p_{\sample}+q-1)$; \{ $\Oh(k)$ time \}\\
$U:=\frag_J(T)$; $V:=\frag_{J'}(P^2)$\;
$Y:=\textsc{Periodic-Periodic-Match}(U,V)$; \{ $\Oh(k^2)$ time \}\\
$Y:=Y \oplus \min(J)$\;
$J':=J' \bmod m$;\\
$X:= \textsc{Pairs-Match}(T,J,P,J')$; \{ $\Oh(k^3)$ time \} \\
\Return{$X\ \cup\ Y$};
\vspace*{0.3cm}
\end{algorithm}

\begin{lemma}\label{lem:per_time}
  After $\Oh(n)$-time preprocessing, algorithm Run-Sample-Matching works in $\Oh(k^3)$ time and
  returns a compact representation that consists of $\Oh(k^3)$ intervals and $\Oh(k^2)$ interval chains, each with difference $q$.
  Moreover, if there is at least one interval chain, then some rotation of the pattern $P$ is $k$-periodic with a $k$-period $\per(\sample)$.
\end{lemma}
\begin{proof}
  See~\cref{alg:periodic}.
  %The fragments $U$ and $V$ are computed in $\Oh(1)$ time.
  The sets $J$ and $J'$ can be computed in $\Oh(k)$ time:
  \begin{claim}
    If $S$ is a string of length $n$, then the sets $\RightMis_k(S,i,j)$ and $\LeftMis_k(S,i,j)$
    can be computed in $\Oh(k)$ time after $\Oh(n)$-time preprocessing.
  \end{claim}
  \begin{proof}
    For $\RightMis_k(S,i,j)$, we use the kangaroo method~\cite{DBLP:journals/tcs/LandauV86,DBLP:journals/tcs/GalilG87}
    to compute the longest common prefix with at most $k$ mismatches of $S[j+1 \dd n-1]$ and $U^\infty$ for $U=S[i \dd j]$.
    The value $\lcp(X^\infty, Y)$ for a fragment $X$ and a suffix $Y$ of a string $S$, occurring at positions $a$ and $b$, respectively, can be computed in constant time as follows. If $\lcp(S[a \dd n-1], S[b \dd n-1]) < |X|$ then we are done. Otherwise the answer is given by $|X|+\lcp(S[b \dd n-1], S[b+|X| \dd n-1])$.
    %The string $U^p$ is not stored explicitly, but all the computations take place cyclically on $U$.
    %This adds $\Oh(p) = \Oh(k)$ to the query complexity. \todo{TK: $U=\sample[0\dd q-1]$ might be a cleaner alternative.}
    The computations for $\LeftMis_k(S,i,j)$ are symmetric.
  \end{proof}
  The $\Oh(k^3)$ and $\Oh(k^2)$ time complexities of computing $X$ and $Y$ follow from Lemmas~\ref{lem:mpp} and~\ref{lem:ppm}, respectively
  (after $\Oh(n)$-time preprocessing).
  The sets $X$ and $Y$ consist of $\Oh(k^3)$ intervals and $\Oh(k^2)$ interval chains, each with difference $q$.

  As for the ``moreover'' statement, by Lemma~\ref{lem:JR}, if any occurrence $q$ is reported in the \textsc{Periodic-Periodic-Match} problem, then it implies
  the existence of $x$ such that $V[x \dd x+m-1]$ is $k$-periodic with a $k$-period $q$.
  However, $V[x \dd x+m-1]$ is a rotation of the pattern $P$.
  This concludes the proof.
\end{proof}

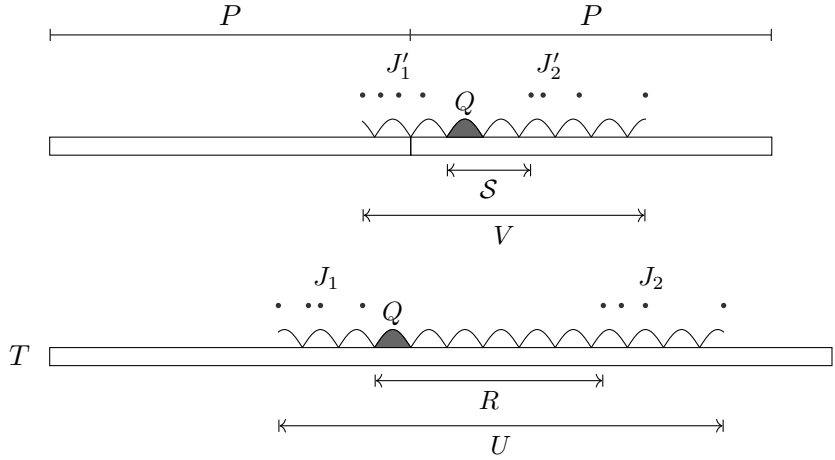
\begin{figure}[thpb]
  \centering
  \input{_fig_pseudocode_notation.tex}
  \caption{Detailed setting in Algorithm~\ref{alg:periodic}; 
$J=J_1 \cup J_2$, $J'=J'_1\cup J'_2$.}\label{fig:app:periodic}
\end{figure}

\noindent
The correctness of the algorithm follows from Lemma~\ref{lem:JR}, as shown in the lemma below.

\begin{lemma}\label{lem:per_corr}
Assume $n\le 2m$. Let $\sample$ be a periodic sample in $P$ with smallest period $q$ and $R$ be an $\sample$-run in $T$.
Let $X$ and $Y$ be defined as in the pseudocode of Run-Sample-Matching.
Then $X \cup Y$ is a set of $k$-occurrences of $P$ in $T$
which is a superset of the solution to \textsc{Sample-Match} for $\sample$ in $R$.
\end{lemma}

\begin{proof}
  Both \textsc{Pair-Match} and \textsc{Periodic-Periodic-Match} problems return positions of $k$-occurrences of $P$ in $T$.
  Assume that $\rot_x(P)$, for $x \ge p_{\sample}+m_{\sample}$, has a $k$-mismatch occurrence in $T$ at position $p$ such that the designated fragment $\sample$ matches a fragment of $R$ exactly.
  Hence, it suffices to show that $p \in X \cup Y$.

  Let $J=\Misp_{k+1}(T,s,s+q-1)$ and $J'=\Misp_{k+1}(P^2,m+p_{\sample},m+p_{\sample}+q-1)$.
  We define $L_1$ and $L_2$ as the subsets of $J$ and $J'$, respectively, that are relevant for this $k$-occurrence, i.e.,
  \[L_1 = J \cap \W_p,\quad L_2 = J' \cap \W_x.\]
  Further let $L'_2 = L_2 \bmod m$.
  If any $i \in L_1$ and $j \in L'_2$ are a matching pair for this $k$-occurrence, then it will be found in the \textsc{Pairs-Match} problem, i.e.\ $p \in X$.
  Let us henceforth consider the opposite case.

	Let $S = T[p \dd p+m-1]$, $S'=\rot_x(P)$, and $i=m-x+p_{\sample}$ be the starting position of $Q=\sample[0 \dd q-1]$ in both strings.
  Further let $I=L_1 \ominus p$ and $I'=L_2 \ominus x$.
  We have that $I \cap I' = \emptyset$ by our assumption that misperiods do not align. We make the following claim.
  
  \begin{claim}
   $\Mis(S,S') = I \cup I'$.% $I$ and $I'$ are the only misperiods in $S$ and $S'$ w.r.t.\ $S[i \dd i+q-1]=Q$ and $S'[i \dd i+q-1]=Q$, respectively.
  \end{claim}
  \begin{proof}
  Note that $I = \Misp_{k+1}(S, i, i+q-1)$ and $I' =\Misp_{k+1}(S', i, i+q-1)$. The latter equality follows from the fact that $\Misp_{k+1}(T, s, s+q-1)=\Misp_{k+1}(T, t, t+q-1)$ for any $t \in [s,s+|R|-q]$.
  We can thus directly apply Lemma~\ref{lem:JR} to strings $S$ and $S'$.
   \end{proof}
 \begin{comment}
    Assume to the contrary that there is a misperiod $t \not\in I$ in $S$.
    Let $J_1=\LeftMis_{k+1}(T,i,i+q-1)$
    and $J_2=\RightMis_{k+1}(T,i,i+q-1)$.
    We have $t < \min J$ (and $|J_1| = k+1$) or $t > \max J$ (and $|J_2| = k+1$).
    The two cases are symmetric, so let us examine the first case.
    
    We have $J_1 \subseteq I$.
    Let $j \ge 1$ be the maximum index such that $|(I \cup I') \cap [j \dd i+q-1]| \ge k+1$ ---such a position exists by our assumptions.
    Then, Lemma~\ref{lem:JR} applied to strings $U=S[j \dd i+q-1]$ and $V=S'[j \dd i+q-1]$ yields a contradiction, because:
    \begin{itemize}
      \item $U =_k V$, since they are fragments occurring at the same position in $S$ and $S'$, respectively.
      \item Let $F=\Misp_{k+1}(U,i',i'+q-1)$ and $F'=\Misp_{k+1}(V,i',i'+q-1)$ for $i'=i-j$.
      Then $F$ and $F'$ are the misperiods from $I$ and $I'$, respectively, restricted to $[j \dd i+q-1]$ are shifted by $j$ positions to the left.
      Hence, $F \cap F' = \emptyset$ by the assumption in this case.
      \item Finally, $|F| + |F'| \ge k+1$ by the definition of the index $j$.
    \end{itemize}

    The proof that there is no misperiod other than the ones from $I'$ in $S'$ is the same as for $I$ and $S$.
    \end{comment}

  In particular, $|I|+|I'| \le k$.
  Moreover,
  $\min(J) < p$ and $p+m-1 < \max(J)$
  as well as
  $\min(J') < x$ and $x+m-1 < \max(J')$, since otherwise we would have $|I| \ge k+1$ or $|I'| \ge k+1$.
  In conclusion, this $k$-occurrence will be found in the \textsc{Periodic-Periodic-Match} problem, i.e.\ $p \in Y$.
\end{proof}

%\subsection{Solution to \textsc{Periodic-Periodic-Match} problem}

The following proposition summarizes the results of this section.

\begin{proposition}\label{prop:nk5}
If $m \le n \le 2m$, $k$-CPM can be solved in $\Oh(n + k^5)$ time and $\cO(n)$ space.
\end{proposition}
\begin{proof}
  We split the pattern into $2k+3$ fragments and choose a sample $\sample$ among them in every possible way.

  If the sample $\sample$ is not periodic, we use the algorithm of Lemma~\ref{lem:nonper} for \textsc{Sample-Match} in $\Oh(k^2)$ time
  (after $\Oh(n)$-time preprocessing).
  It returns a representation of $k$-occurrences as a union of $\Oh(k^2)$ intervals.

  If the sample $\sample$ is periodic, we need to find all $\sample$-runs in $T$.
  By Lemma~\ref{lem:P1_runs}, there are $\Oh(k)$ of them and they can all be computed in $\Oh(k)$ time
  (after $\Oh(n)$-time preprocessing).
  For every such $\sample$-run $R$, we apply the Run-Sample-Matching algorithm.
  Its correctness follows from Lemma~\ref{lem:per_corr}.
  By Lemma~\ref{lem:per_time}, it takes $\Oh(k^3)$ time and returns $\Oh(k^3)$ intervals and $\Oh(k^2)$ interval chains, each with difference $\per(\sample)$, of $k$-occurrences of $P$ in $T$
  (after $\Oh(n)$-time preprocessing).
  Over all $\sample$-runs, this takes $\Oh(k^4)$ time after the preprocessing and returns $\Oh(k^4)$ intervals and $\Oh(k^3)$ interval chains.
  
  By Lemma~\ref{lem:per_time}, if any interval chains are reported in Run-Sample-Matching, then some rotation of the pattern is $k$-periodic with a $k$-period $\per(\sample)$.
  Then, at least $k+2$ of the $2k+3$ pattern fragments do not contain misperiods and hence they must have a period $q=\per(\sample)$.
  This is actually their smallest period, for if one of these fragments $\sample'$ had a period $q' < q$, then
  $|\sample'| \ge |\sample|-1$ and, by Fine and Wilf's periodicity lemma~\cite{fine1965uniqueness}, 
  $\sample'$ would have a period $q''=\gcd(q,q')<q$, which would imply that $Q$ would also have a period $q''$ and hence $\sample$ as well.
  Thus, throughout the course of the algorithm, Run-Sample-Matching can only return interval chains of period $\per(\sample)$ by the pigeonhole principle.
   %\todo[inline]{ Note that if \textsc{Periodic-Periodic-Match} for some sample $\sample$ returns any occurrences then this implies that some rotation of $P$ is at distance at most $k$ from a substring of $\sample[1 \dd \per(\sample)]^\infty$.
  
  In total, \textsc{Sample-Match} takes $\Oh(k^4)$ time for a given sample (after preprocessing), $\Oh(n+k^5)$ time in total,
  and returns $\Oh(k^5)$ intervals and $\Oh(k^4)$ interval chains of $k$-occurrences, each with the same difference $q$.
  Let us note that an interval is a special case of an interval chain with difference, say, 1.
  We then apply Lemma~\ref{lem:union_chains} to compute the union of all chains of occurrences and the union of all intervals in $\Oh(n + k^5)$ total time.
  In the end we return the union of the two unions.

  In order to bound the space required by our algorithm by $\cO(n)$, we do not store each interval chain explicitly throughout the execution of the algorithm. Instead, for each interval chain we increment/decrement a constant number of cells in a ($\G_q$-shaped for interval chains or $\G_1$-shaped for intervals) 2D array of size $\cO(n)$ as in the proof of~\cref{lem:union_chains}, and compute the union of all such interval chains in the end.
\end{proof}

%\subsection{Wrapping up}%$\cO(n+\frac{n}{m}\,{\small k^5})$-time algorithm}
\section{An $\cO(n+\frac{n}{m}\,{k^4})$-time Algorithm}
Let us observe that for each non-periodic fragment $\sample$ we have to solve $\cO(k)$ instances of \textsc{Pair-Match}, while for each periodic fragment $\sample$ and each $\sample$-run, 
we obtain two sets $J$ and $J'$, each of cardinality $\cO(k)$, where each pair of elements in $J \times J'$ requires us to solve an instance of \textsc{Pair-Match}.
Further recall that each instance of \textsc{Pair-Match} reduces to two calls to our $\cO(k)$-time algorithm for \textsc{Anchor-Match}.
We thus consider $\cO(k^4)$ \textsc{Anchor-Match} instances in total, yielding a total time complexity of $\cO(k^5)$. 
As can be seen in the proof of~\cref{prop:nk5} and the pseudocode, this is the bottleneck of our algorithm, with everything else requiring $\cO(n+k^4)$ time.
We will decrease the number of calls to \textsc{Pair-Match} by using a marking trick.

%\subsection{Removing a $O(k)$-Factor in time complexity}
%
We first present a simple application of the marking trick.
Suppose that we are in the standard $k$-mismatch problem, where we are to find all $k$-occurrences (not circular ones) of a pattern $P$ of length $m$ in a text $T$ of length $n$, and $n \leq 2m$.
Further suppose that $P$ is square-free, or, in other words, that it is nowhere periodic.
Let us consider the following algorithm. We split the pattern into $k+1$ fragments of length roughly $m/k$ each. 
Then, at each $k$-occurrence of $P$ in $T$, at least one of the $k+1$ fragments must match exactly. 
We then find the $\cO(k)$ such exact matches of each fragment in $T$ and each of them nominates a position for a possible $k$-occurrence of $P$. 
We thus have $\cO(k^2)$ candidate positions in total to verify.

Now consider the following refinement of this algorithm. We split the pattern into $2k$ fragments (instead of $k+1$), each of length roughly $m/(2k)$.
Then, at each $k$-occurrence of $P$ in $T$, at least $k$ of the fragments must match exactly; we exploit this fact as follows. 
Each exact occurrence of a fragment in $T$ gives a mark to the corresponding position for a $k$-occurrence of $P$.
There are thus $\cO(k^2)$ marks given in total. However, we only need to verify positions with at least $k$ marks and these are now $\cO(k)$ in total.
An illustration is provided in~\cref{fig:marking}.

\begin{figure}[ht]
  \centering
  \input{_fig_marking.tex}
  \caption{We consider $k=4$. To the left: a candidate starting position for $P$ given by an exact match of one of the $5$ fragments. To the right: a candidate starting position for $P$ given by exact matches of $4$ of the $8$ fragments.}\label{fig:marking}
\end{figure}
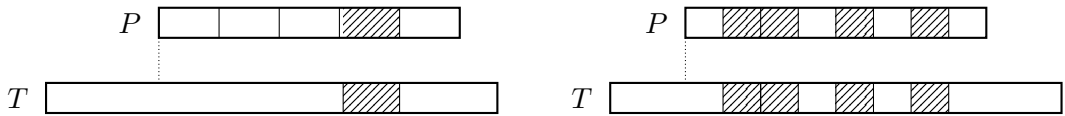
%We will use this in a marking scheme, which we now describe informally. Each relevant exact occurrence of a fragment in $T$ gives two possible anchors in $T$ for a $k$-occurrence that matches this fragment exactly; we give a mark to each of these anchors. In the end, only anchors with at least $k$ marks need to be considered.

Let us get back to the $k$-CPM problem.
%We refine the algorithm as follows. We split the pattern in $2k+1$ fragments (instead of $k+2$) each of length roughly $m/(2k+1)$.
Recall that each $k$-occurrence implies that at least $k+2$ fragments match exactly.
We run the algorithm yielding \cref{prop:nk5} with a single difference.  
Instead of processing each instance of \textsc{Pair-Match} separately, we apply the marking trick in order to decrease the exponent of $k$ by one.
This is achieved by a reduction in the number of calls to the algorithm that solves \textsc{Anchor-Match}.
For each of the $\cO(k^4)$ instances of \textsc{Pair-Match} we mark the two possible anchors for a $k$-occurrence and note that only anchors with at least $k+2$ marks need to be verified; these are $\cO(k^3)$ in total. 
Finally, for each such anchor we apply our solution to the \textsc{Anchor-Match} problem, which requires $\cO(k)$ time, hence obtaining an $\cO(n+k^4)$-time algorithm.

The correctness of this approach follows from the following fact. 
(Note that \textsc{Periodic-Periodic-Match} is not affected by this modification.)
If there is a $k$-occurrence at position $i$ of $T$ that has not been returned by any of the calls to \textsc{Periodic-Periodic-Match}, 
then, the algorithm would return $i$ through at least $k+2$ calls to \textsc{Anchor-Match}. We arrive at the following result.

\begin{proposition}\label{prop:nk4}
If $m \le n \le 2m$, $k$-CPM can be solved in $\Oh(n + k^4)$ time and $\cO(n)$ space.
\end{proposition}

%\subsection{Main Result: Long Texts}
%\section{Final Result}

Both Propositions~\ref{prop:nk} and~\ref{prop:nk4} use $\cO(n)$ space. Moreover, Proposition~\ref{prop:nk4} assumes that $n \le 2m$.
In order to solve the general version of the $k$-CPM problem, where $n$ is arbitrarily large, efficiently and using $\cO(m)$ space,
we use the so-called standard trick: we split the text into $\cO(n/m)$ fragments, each of length $2m$
(perhaps apart from the last one), starting at positions equal to $0 \bmod m$.

We need, however, to ensure that the data structures for answering $\lcp$, $\lcsuf$, and other internal queries over each 
such fragment of the text can be constructed in $\cO(m)$ time when the input alphabet $\Sigma$ is large.
As a preprocessing step we hash the letters of the pattern using perfect hashing.
For each key, we assign a unique identifier from $\{1,\ldots,m\}$. 
This takes $\cO(m)$ (expected) time and space~\cite{Fredman:1984:SST:828.1884}. When reading a fragment $F$ of length (at most) $2m$ of the text we 
look up its letters using the hash table. If a letter is in the hash table we replace it in $F$ by its rank value; otherwise we replace it by rank $m+1$. 
We can now construct the data structures in $\cO(m)$ time and thus our algorithms can be implemented in $\cO(m)$ space. 

If $\Sigma=\{1,\ldots,n^{\cO(1)}\}$, the same bounds can be achieved deterministically.
Specifically, we consider two cases. If $m > \sqrt{n}$ we sort the letters of every text fragment and of the pattern in $\cO(m)$ time per fragment because $n$ is polynomial in $m$ and $|\Sigma|$ is polynomial in $n$. Then we can merge the two sorted lists and replace the letters in the pattern and the text fragments by their ranks.
Otherwise ($m \leq \sqrt{n}$), we construct a deterministic dictionary for the letters of the pattern in $\cO(m\log^2\log m)$ time~\cite{DBLP:conf/icalp/Ruzic08}. The dictionary uses $\cO(m)$ space and answers queries in constant time; we use it instead of perfect hashing in the previous solution.

We combine Propositions~\ref{prop:nk} and~\ref{prop:nk4} with the above discussion to get our final result.

\begin{theorem}
Circular Pattern Matching with $k$ Mismatches can be solved in $\cO(\min(nk,\; n+\frac{n}{m}\,{\small k^4}))$ time and $\cO(m)$ space.
\end{theorem}

Our algorithms output all positions in the text where some rotation of the pattern occurs with $k$ mismatches. It is not difficult to extend the algorithms to output, for each of these positions, a corresponding rotation of the pattern.

\bibliographystyle{elsarticle-num}
\bibliography{k-mismatches}

\end{document}

%% file: _fig_circular_occ_ex.tex
\begin{tikzpicture}
  \tikzstyle{red}=[color=red!90!black]
  \tikzstyle{darkred}=[color=red!50!black]
  \tikzstyle{blue}=[color=blue!50!black]
  \tikzstyle{black}=[color=black]
  \definecolor{turq}{RGB}{74,223,208}
  \definecolor{pink}{RGB}{254,193,203}

  \begin{scope}[xshift=-3.5cm,yshift=-1cm]
  \node at (-0.3,0) [left, above] {$P=$};
  \draw [fill=pink] (0.16, 0.07) rectangle (0.75, 0.38);
  \draw [fill=turq] (0.75, 0.07) rectangle (1.94, 0.38);
  \foreach \c/\s [count=\i from 0] in {a/darkred,a/darkred,b/blue,b/blue,b/blue,b/blue} {
    \node at (\i * 0.3 + 0.3, 0) [above, \s] {\tt \c};
    \node at (\i * 0.3 + 0.3, 0.1) [below] {\tiny \i};
  }
  \draw (2 * 0.3 + 0.3, -0.08) node[circle, minimum width=0.25cm,inner sep=0.5mm,draw] {}
    node [below right] {\tiny split point=2};
  \end{scope}
  
  \begin{scope}[yshift=-1cm]
  \node at (-0.3,0) [left, above] {$T=$};
  \draw [fill=pink] (2.54, 0.07) rectangle (3.13, 0.38);
  \draw [fill=turq] (1.35, 0.07) rectangle (2.54, 0.38);
  \foreach \c/\s [count=\i from 0] in {a/black,a/black,c/black,c/black,b/blue,b/blue,x/red,b/blue,a/darkred,a/darkred,a/black,b/black} {
    \node at (\i * 0.3 + 0.3, 0) [above, \s] {\tt \c};
    \node at (\i * 0.3 + 0.3, 0.1) [below] {\tiny \i};
  }
  \draw (8 * 0.3 + 0.3, -0.08) node[circle, minimum width=0.25cm,inner sep=0.5mm,draw] {}
    node [below right] {\tiny anchor=8};
  \end{scope}

  \begin{scope}[yshift=-2cm, xshift=1.2cm]
  \node at (-0.7,-0.1) [above] {$\rot_2(P)=$};
  \draw [fill=pink] (1.34, 0.07) rectangle (1.93, 0.38);
  \draw [fill=turq] (0.15, 0.07) rectangle (1.34, 0.38);
  \foreach \c/\s [count=\i from 0] in {b/blue,b/blue,b/blue,b/blue,a/darkred,a/darkred} {
    \node at (\i * 0.3 + 0.3, 0) [above, \s] {\tt \c};
  }
  \foreach \ii [count=\i from 0] in {2,3,4,5,0,1} {
    \node at (\i * 0.3 + 0.3, 0.1) [below] {\tiny \ii};
  }
  \end{scope}

  \end{tikzpicture}

%% file: _fig_light_fragments_ex.tex
\begin{tikzpicture}
  \tikzstyle{red}=[color=red!90!black]
  \tikzstyle{darkred}=[color=red!50!black]
  \tikzstyle{blue}=[color=blue!50!black]
  \tikzstyle{black}=[color=black]
  \definecolor{turq}{RGB}{74,223,208}
  \definecolor{pink}{RGB}{254,193,203}
  \definecolor{pink-orange}{rgb}{1.0, 0.6, 0.4}
  
 %%%%%%%%%%%%%% P left %%%%%%%%%%%%%%%
  
  \begin{scope}
  \node at (-0.3,0) [left, above] {$P^2=$};
  \draw [fill=turq] (0.15, 0.07) rectangle (4.95, 0.38);
  \foreach \c [count=\i from 0] in {a,b,a,a,b,a,b,a,a,b,a,a,b,a,b,a} {
    \node at (\i * 0.3 + 0.3, 0) [above] {\tt \c};
  }
  \end{scope}
  
  %%%%%%%%%%%%%% P right %%%%%%%%%%%%%%%
  
  \begin{scope}[xshift=4.8cm]
  \draw [fill=pink] (0.15, 0.07) rectangle (4.95, 0.38);
  \foreach \c [count=\i from 0] in {a,b,a,a,b,a,b,a,a,b,a,a,b,a,b,a} {
    \node at (\i * 0.3 + 0.3, 0) [above] {\tt \c};
  }
  \end{scope}
  
 %%%%%%%%%%%%%% T %%%%%%%%%%%%%%%
  
  \begin{scope}[yshift=-0.5cm]
  \node at (-0.3,0) [left, above] {$T=$};
  \foreach \c [count=\i from 0] in {b,b,a,a,b,a,a,a,a,b,a,a,a,a,b,b,a,b,a,b,b,a,b,a,b,b,a,a,b,a,a,b} {
    \node at (\i * 0.3 + 0.3, 0) [above] {\tt \c};
  }
  \end{scope}
  
  %%%%%%%%%%%%%% V %%%%%%%%%%%%%%%
  
  \begin{scope}[yshift=-1cm]
  \node at (-0.3,0) [left, above] {$V=$};
  \draw [fill=pink-orange] (0.45, -0.21) rectangle (1.35, 0.05);
  \draw [fill=pink-orange] (2.25, -0.21) rectangle (2.85, 0.05);
  \draw [fill=pink-orange] (4.05, -0.21) rectangle (4.65, 0.05);
  \draw (2.25, 0.08) rectangle (7.05, 0.38);
  \foreach \c [count=\i from 0] in {,0,0,0,0,0,1,0,0,0,0,0,1,0,0,1,0,0,0,1,0,0,0,0,1,0,0,0,0,0,1,} {
    \node at (\i * 0.3 + 0.3, 0) [above] {\tt \c};
    \node at (\i * 0.3 + 0.3, 0.1) [below] {\tiny \i};
  }
  \draw (16 * 0.3 + 0.3, -0.08) node[circle, minimum width=0.28cm,inner sep=0.5mm,draw] {}
    node [below right] {\tiny anchor=16};
  \end{scope}

\end{tikzpicture}

%% file: _fig_int_chains.tex
\begin{tikzpicture}
  \tikzstyle{red}=[color=red!90!black]
  \tikzstyle{darkred}=[color=red!50!black]
  \tikzstyle{blue}=[color=blue!50!black]
  \tikzstyle{black}=[color=black]
  \definecolor{turq}{RGB}{74,223,208}
  \definecolor{pink}{RGB}{254,193,203}
  \definecolor{pink-orange}{rgb}{1.0, 0.6, 0.4}

 %%%%%%%%%%%%%% U %%%%%%%%%%%%%%%
  
  \begin{scope}
  \node at (-0.3,0) [left, above] {$U=$};
  \foreach \c [count=\i from 0] in {0,1,0,0,0,0,0,0,0,0,0,0,0,0,0,0} {
    \node at (\i * 0.3 + 0.3, 0) [above] {\tt \c};
  }
  \end{scope}
  
  %%%%%%%%%%%%%% V %%%%%%%%%%%%%%%
  
  \begin{scope}[yshift=-0.5cm]
  \draw [fill=pink-orange] (0.15, -0.21) rectangle (0.75, 0.05);
  \draw [fill=pink-orange] (1.35, -0.21) rectangle (1.95, 0.05);
  \draw [fill=pink-orange] (2.55, -0.21) rectangle (3.15, 0.05);
  \draw [fill=pink-orange] (3.75, -0.21) rectangle (4.35, 0.05);
  \node at (-0.3,0) [left, above] {$V=$};
  \foreach \c [count=\i from 0] in {0,0,0,0,0,0,0,0,0,0,0,0,0,0,1,0,0,0,0,0,0,0,0,0,0,0,0,0,0} {
    \node at (\i * 0.3 + 0.3, 0) [above] {\tt \c};
    \node at (\i * 0.3 + 0.3, 0.1) [below] {\tiny \i};
  }

  \end{scope}

\end{tikzpicture}

%% file: _fig_bring.tex
\begin{tikzpicture}[scale=0.5]

%%% S %%%

\begin{scope}[yshift=3.5cm]

  \draw[thick] (0,-0.4) rectangle (20,0.4);
  \draw[fill=white!75!black,draw] (7,-0.4) rectangle (10,0.4);

  %%labels%%
  \draw (0,0) node[left=0.1cm] {$S$};
  \draw (8.5,0) node {$Q$};
    \draw[red] (2,0) node {$*$};
  \draw[red] (11,0) node {$*$};
    \draw[red] (18,0) node {$*$};
\end{scope}

%%% S' %%%

   \draw[thick] (0,-0.4) rectangle (20,0.4);
  \draw[fill=white!75!black,draw] (7,-0.4) rectangle (10,0.4);
   
  %%labels%%
  \draw (0,0) node[left=0.1cm] {$S'$};
  \draw (7.25,-0.4) node[below] {$i$};
  \draw (9.75,-0.4) node[below] {$j$};
  \draw (8.5,0) node {$Q$};
  \draw[red] (3.5,0) node {$*$};
  \draw[red] (15,0) node {$*$};

%%% periodic %%%

\begin{scope}[yshift=1.5cm]
  \draw (0,-0.4) rectangle (20,0.4);
  \draw[fill=white!75!black,draw] (7,-0.4) rectangle (10,0.4);
  \draw (0,0) node[left=0.1cm] {$X$};
  \draw (8.5,0) node {$Q$};
    
        %%period%%
\begin{scope}[yshift=0.6cm]
  \foreach \x in {-2,1,4,7,10,13,16,19}{
    \clip (0,0) rectangle (20,0.7);
    \draw[xshift=\x cm] (0,0) sin (1.5,0.4) cos (3,0);
  }
  
\end{scope}

\end{scope}

\end{tikzpicture}

%% file: _fig_primitive.tex
\begin{tikzpicture}[scale=0.8]

    \begin{scope}[yshift=2cm]
      \draw[thin] (0,0) rectangle (3,0.3);
      \draw[thin, fill=black!20!white] (3,0) rectangle (9,0.3)
        node [above, right, yshift=0.2cm, xshift=-0.2cm] {$V$};
      \draw[thin] (9,0) rectangle (11,0.3);
      \draw (0,0.15) node [left] {$P^2$};
 
      \node (ii) [circle,fill=black!80!white,inner sep=0.25mm] at (4, 0.15) {};
      \node [below of=ii, yshift=0.55cm] {$i$};
   
      \node (zz) [circle,fill=black!80!white,inner sep=0.25mm] at (5, 0.15) {};
      \node [below of=zz, yshift=0.5cm] {$x$};
    
      \draw[thin, fill=black!30!white] ($(zz)+(0, 0.3)$) rectangle ($(zz)+(3, 0.6)$);
      \draw[|<->|] ($(zz)+(0, 0.8)$) -- ($(zz)+(3, 0.8)$) node[midway, above] {\small $m$};
    \end{scope}
    
    \begin{scope}[yshift=0cm]
      \draw[thin] (0,0) rectangle (2,0.3);
      \draw[thin, fill=black!20!white] (2,0) rectangle (10,0.3)
        node [below, right, yshift=0.2cm, xshift=-0.2cm] {$U$};
      \draw[thin] (10,0) rectangle (12,0.3);
       \draw (0,0.15) node [left] {$T$};

      \node (ii1) [circle,fill=black!80!white,inner sep=0.25mm] at (4.5, 0.15) {};
      \node [above of=ii1, yshift=-0.5cm] {$i'$};
   
      \node (zz) [circle,fill=black!80!white,inner sep=0.25mm] at (5.5, 0.15) {};
      \node [above of=zz, yshift=-0.54cm] {$p$};
    
      \draw[thin, fill=black!30!white] ($(zz)+(0, -0.6)$) rectangle ($(zz)+(3, -0.3)$);
      \draw[|<->|] ($(zz)+(0, -0.8)$) -- ($(zz)+(3, -0.8)$) node[midway, below] {\small $m$};
    \end{scope}
  
\end{tikzpicture}

%% file: _fig_pseudocode_notation.tex
\begin{tikzpicture}[scale=0.8]
  \def\perLen{0.6}

\begin{scope}[yshift=3.5cm]
  \draw[thin] [|-|] (0,2) -- (6, 2) node[midway,above] {$P$};
  \draw[thin] [-|] (6,2) -- (12, 2) node[midway,above] {$P$};
  \draw[thin] (0,0) rectangle (6,0.3);
  \draw[thin] (6,0) rectangle (12,0.3);
  \foreach \x in {5.2, 5.5, 5.8, 6.2, 8.0, 8.2, 8.8, 9.9} {
    \node[circle,fill=black!80!white,inner sep=0.25mm] at (\x, 1.0) {};
  }
  \node at (5.8, 1.1) [above] {\small $J'_1$};
  \node at (8.3, 1.1) [above] {\small $J'_2$};
  \node at (6.9, 0.5) [above] {\small $Q$};
  
  \draw [|<->|] (11*\perLen, -0.25) -- (8.0, -0.25)
    node[midway, below] {\small $\sample$};
  \draw [|<->|] (5.2, -1.0) -- (9.9, -1.0)
    node[midway, below] {\small $V$};
\end{scope}

\begin{scope}[yshift=3.5cm]
    \clip (5.2, 0.3) rectangle (9.9, 1.5);
    \foreach \x in {0,...,50}{
      \draw (\x*\perLen,0.3) sin (\x*\perLen+0.5*\perLen,0.6) cos (\x*\perLen+\perLen,0.3);
    }
     \draw[fill=black!60!white] (11*\perLen,0.3) sin (11*\perLen+0.5*\perLen,0.6) cos (11*\perLen+\perLen,0.3) -- cycle;
\end{scope}

\begin{scope}[yshift=0cm]
  \draw[thin] (0,0) rectangle (13,0.3);
  \node at (-0.5, 0.18) {$T$};
  \node at (9*\perLen+\perLen*0.5, 0.5) [above] {\small $Q$};
  % \draw (9*\perLen,0)--+(0,0.3) node [below, yshift=-0.1cm] {\small $s$};
  \node at (4.6, 1.1) [above] {\small $J_1$};
  \node at (10, 1.1) [above] {\small $J_2$};
  \foreach \x in {3.8, 4.5, 4.3, 5.2, 9.2, 9.5, 9.9, 11.2} {
    \node[circle,fill=black!80!white,inner sep=0.25mm] at (\x, 1.0) {};
  }
  \draw [|<->|] (9*\perLen, -0.25) -- (9.2, -0.25)
    node[midway, below] {\small $R$};
  \draw [|<->|] (3.8, -1) -- (11.2, -1)
    node[midway, below] {\small $U$};

\end{scope}

\begin{scope}[yshift=0cm]
    \clip (3.8, 0.3) rectangle (11.2, 1.5);
    \foreach \x in {0,...,50}{
      \draw (\x*\perLen,0.3) sin (\x*\perLen+0.5*\perLen,0.6) cos (\x*\perLen+\perLen,0.3);
    }
    \draw[fill=black!60!white] (9*\perLen,0.3) sin (9*\perLen+0.5*\perLen,0.6) cos (9*\perLen+\perLen,0.3) -- cycle;
\end{scope}

\end{tikzpicture}

%% file: _fig_marking.tex
\begin{tikzpicture}[scale=0.5]

%%% LEFT %%%%%%%%%%%%%%%%%%%%%%%%

%%% P %%%%%%%%%%%%%%%%%%%%%%%%

\draw[thick] (3,-0.4) rectangle (11,0.4);
\draw (3,0) node[left=0.1cm] {$P$};

\foreach \x in {1,2,3,4}{
    \draw[xshift=1.6*\x cm] (3,-0.4) -- (3,0.4);
  }
  
\fill[pattern=north east lines] (7.9,-0.4) rectangle (9.4,0.4);

\draw[densely dotted] (3,-0.4) -- (3,-1.6);

%%% T %%%%%%%%%%%%%%%%%%%%%%%%
	
\begin{scope}[yshift=-2cm]
   \draw[thick] (0,-0.4) rectangle (12,0.4);
   \draw (0,0) node[left=0.1cm] {$T$};
  
\draw[pattern=north east lines] (7.9,-0.4) rectangle (9.4,0.4);
\end{scope}

%%% RIGHT %%%%%%%%%%%%%%%%%%%%%%%%

\begin{scope}[xshift=15cm]

%%% P %%%%%%%%%%%%%%%%%%%%%%%%

\draw[thick] (2,-0.4) rectangle (10,0.4);
\draw (2,0) node[left=0.1cm] {$P$};

\foreach \x in {1,2,3,4,5,6,7}{
    \draw[xshift=\x cm] (2,-0.4) -- (2,0.4);
  }
  
\foreach \x in {1,2,4,6}{
	\fill[xshift=\x cm, pattern=north east lines] (2,-0.4) rectangle (3,0.4);
	\draw[xshift=\x cm, pattern=north east lines] (2,-1.6) rectangle (3,-2.4);
}

\draw[densely dotted] (2,-0.4) -- (2,-1.6);

%%% T %%%%%%%%%%%%%%%%%%%%%%%%
	
\begin{scope}[yshift=-2cm]
   \draw[thick] (0,-0.4) rectangle (12,0.4);
   \draw (0,0) node[left=0.1cm] {$T$};
  
\end{scope}
\end{scope}

\end{tikzpicture}